\documentclass[centertags,12pt]{amsart}
\usepackage{amssymb}
\usepackage{latexsym}
\usepackage[pdftex]{color,graphicx}

\makeatletter
\renewcommand{\subsection}{\@startsection
{subsection}{2}{0mm}{\baselineskip}{-0.25cm}
{\normalfont\normalsize\em}}
\makeatother

\newtheorem{theorem}{Theorem}[section]
\newtheorem{proposition}[theorem]{Proposition}
\newtheorem{corollary}[theorem]{Corollary}
\newtheorem{lemma}[theorem]{Lemma}
{\theoremstyle{definition}

\newtheorem{example}[theorem]{Example}

}
\theoremstyle{remark}
\newtheorem{remark}[theorem]{Remark}

\newcommand{\bba}{{\bf a}}
\newcommand{\bbb}{{\bf b}}
\newcommand{\bbc}{{\bf c}}
\newcommand{\bbv}{{\bf v}}
\newcommand{\bbx}{{\bf x}}
\newcommand{\bby}{{\bf y}}
\newcommand{\bbz}{{\bf z}}
\newcommand{\bbuno}{{\bf 1}}
\newcommand{\bbcero}{{\bf 0}}
\newcommand{\fq}{\mathbb{F}_q}
\newcommand{\fqd}{\mathbb{F}_{q^2}}
\newcommand{\fqr}{\mathbb{F}_{q^r}}
\newcommand{\fqdr}{\mathbb{F}_{q^{2r}}}
\newcommand{\cB}{\mathcal B}
\newcommand{\cH}{\mathcal H}
\newcommand{\cL}{\mathcal L}
\newcommand{\cV}{\mathcal V}
\newcommand{\cX}{\mathcal X}
\newcommand{\ev}{\mbox{\rm ev}}
\newcommand{\car}{\mbox{\rm char}}
\newcommand{\degg}{\mbox{\rm deg}}
\newcommand{\divv}{\mbox{\rm div}}
\newcommand{\res}{\mbox{\rm res}}
\newcommand{\supp}{\mbox{\rm supp}}
\newcommand{\tr}{\mbox{\rm tr}}
\newcommand{\wt}{\mbox{\rm wt}}

\begin{document}

\title[Quantum codes from Algebraic Geometry codes of Castle type]{Quantum error-correcting codes from Algebraic Geometry codes of Castle type
}

\author{Carlos Munuera} 
\address{Department of Applied Mathematics, University of Valladolid, Avda Salamanca SN, 47014 Valladolid, Castilla, Spain}
   \email{cmunuera@arq.uva.es}

\author{Wanderson Ten\'orio} 
 \address{University of Campinas (UNICAMP), Institute of Mathematics, 
Statistics and Computer Science (IMECC), R. 
S\'ergio Buarque de Holanda, 651, Cidade 
Universit\'aria \lq\lq Zeferino Vaz", 13083-059, Campinas, SP, Brazil}
  \email{dersonwt@yahoo.com.br}

\author{Fernando Torres}
  \address{University of Campinas (UNICAMP), Institute of Mathematics, 
Statistics and Computer Science (IMECC), R. 
S\'ergio Buarque de Holanda, 651, Cidade 
Universit\'aria \lq\lq Zeferino Vaz", 13083-059, Campinas, SP, Brazil}
  \email{ftorres@ime.unicamp.br}

\begin{abstract}
We study Algebraic Geometry codes producing quantum error-correcting  codes by the CSS construction. We pay particular attention to the family of  Castle codes.  We show that many of the examples known in the literature in fact belong to this family of codes.   We systematize these constructions by showing the common theory that underlies all of them.
\end{abstract}
\keywords{Quantum codes, Self-orthogonal codes,  Algebraic geometry codes,  Castle codes, Trace codes.\\
The final publication is available at Springer via http://dx.doi.org/10.1007/s11128-016-1378-9.}

\maketitle

\section{Introduction}
\label{intro}

Quantum error-correcting codes are essential to protect quantum information. In the last years much research has been done  to find good quantum codes following several ways. 
Calderbank and Shor \cite{CS}, and Steane \cite{Ste}, showed that quantum codes can be  derived from classical linear error-correcting codes verifying certain self-orthogonality properties  \cite{AK}, including Euclidean and Hermitian self-orthogonality. This method, known as  {\em CSS construction}, 
has received a lot of interest and it has allowed to find many powerful quantum {\em stabilizer} codes. 

Among all the classical codes used to produce quantum stabilizer codes, Algebraic Geometry (AG) codes have received considerable attention \cite{Chen,GH,Ji,JX,KM,KW,SK,Sh}. Conditions for Euclidean self-orthogonality of AG codes are well known \cite{SD}  and allow us to translate the combinatorial nature  of this problem into geometrical terms concerning the arithmetic of the involved curves. Furthermore, Hermitian self-orthogonality can be easily ensured in a similar manner.

From among all curves  used to get AG codes, we can highlight the family of Castle and weak Castle curves \cite{Castle2}, that combine the good properties of having a reasonable simple handling and giving codes with excellent parameters. In fact, most of the one-point AG codes studied in the literature belong to the family of Castle codes. Besides these codes have, in a a natural way, self-orthogonality properties which  are very close to those required for obtaining quantum stabilizer codes. It follows from the foregoing  that many of the AG codes used to derive quantum codes are particular cases of weak Castle codes. 
In this article we systematize these constructions, including them in the overall framework of Castle codes.  To this end, we show the common theory that underlies all of them and  we include many examples, some of which refer to curves and codes already treated in the literature.

The paper is organized as follows. In Section 2 we overview the constructions of quantum codes from classical codes. Section 3 is devoted to Algebraic Geometry codes, and mainly to Castle codes. These provide sequences of self-orthogonal and formally self-orthogonal codes that can be used to produce quantum codes.  We give sufficient conditions for self-duality  and  study in detail some particular families of curves and codes. In Section 4 we show some parameters of quantum codes we have obtained from the curves presented in Section 3. Rather than obtaining codes with new or excellent parameters, we are interested in showing the common bases on which all of them are based. Finally in Section 5 we consider trace codes of AG codes defined over extensions of $\fq$ and the quantum codes derived from them. Trace codes are closely related to subfield subcodes, from which have recently been obtained quantum codes with excellent parameters \cite{GH,GHR}.

As a notation, given a finite field $\fq$, we write $\fq^{\times}=\fq\setminus\{0\}$. Given a vector $\bbx\in\fq^n$ and an integer $t$, $\bbx^t=(x_1^t,..,x_n^t)$ (when this makes sense). If $X\subset\fq^n$, then $X^t=\{\bbx^t : \bbx\in X\}$.

\section{Preliminaries on Quantum  error-correcting codes}
\label{sec:1}

Let $q$ be a prime power.
A $q$-ary {\em quantum  error-correcting code} of {\em length} $n$ and {\em dimension} $K\ge 1$ is a $K$-dimensional linear subspace $Q$ of $\mathbb{C}^{q^n}$, where $\mathbb{C}$ is the field of complex numbers. If $Q$ can detect $d-1$ errors and correct $\lfloor d-1\rfloor/2$ errors, then it is usually referred as a quantum $[[n,k,d]]_q$  code, where $k=\log_q(K)$.  

$q$-ary quantum codes can be obtained from classical codes $C$ over $\mathbb{F}_q$ and over $\mathbb{F}_{q^2}$ by the so-called CSS construction, as  explained below. Quantum codes constructed in this way are {\em stabilizer} quantum codes. 
In this paper we are mainly interested in quantum codes obtained from classical  algebraic geometry codes. For all facts concerning quantum codes and how to derive them from classical codes, we refer to \cite{AK,KM}.

\subsection{Quantum codes from linear codes over $\mathbb{F}_q$}
\label{sec:2}

We shall denote by $\langle -,-\rangle$ the usual (Euclidean) inner product in $\mathbb{F}_q^n$,
 $\langle \bba,\bbb \rangle=\sum a_ib_i$.  Given a linear code $C$ over $\fq$, the {\em dual} of $C$ is the code
\begin{equation}
C^{\perp}=\{ \bbv\in\fq^n \ : \ \langle \bbv,\bbc \rangle=0 \mbox{ for all $\bbc\in C$} \}.
\end{equation}
$C$ is called {\em self-dual} if $C=C^{\perp}$, and 
{\em self-orthogonal} if $C\subseteq C^{\perp}$, that is if $\langle \bba,\bbb \rangle=0$ for all $\bba,\bbb\in C$.
The so-called {\em CSS code construction} allows to obtain quantum codes from classical codes over $\fq$ as follows \cite[Cor. 2.15]{KM}.

\begin{theorem}\label{qfq}
Let $C_1,C_2$ be two linear codes over $\fq$ of length $n$ and dimensions $k_1$ and $k_2$ respectively with   $C_1 \subseteq C_2$. Then there exists a $[[n,k_2-k_1,d]]_q$  code with minimum distance
$d=\min\{ \wt(\bbc) \ : \ \bbc\in(C_2\setminus C_1)\cup(C_1^{\perp}\setminus C_2^{\perp}) \}$.
\end{theorem} 

\begin{corollary}
Let $C$ be a self-orthogonal $[n,k,d]$ code over $\fq$. Then there exists a $[[n,n-2k, d]]_q$  code with minimum distance
$d=\min\{ \wt(\bbc) \ : \ \bbc\in(C^{\perp}\setminus C) \}\ge d(C^{\perp})$.
\end{corollary}

This Corollary has been extensively used to ensure the existence of many quantum codes. In \cite{JX}  Jin and Xing showed that for $q$ even,  the orthogonality condition can be weakened in the sense that we can obtain self-orthogonal codes from other codes "close" to be self-orthogonal.   
To be more precise, let us denote by $*$ the coordinatewise multiplication in $\fq^n$, $\bba*\bbb=(a_1b_1,\dots,a_nb_n)$.
Given an $n$-tuple $\bbx$ of nonzero elements in $\fq$, the map $\fq^n\rightarrow\fq^n$, $\bba\mapsto\bbx*\bba$, is linear and bijective. Furthermore it is an isometry for the Hamming metric. For a linear code $C$ over $\fq$, we shall write
$\bbx*C=\{ \bbx*\bbc \ : \   \bbc\in C\}$.
Two codes $C_1$ and $C_2$ over $\fq$ are called {\em formally equivalent} (or {\em twisted} according other authors) if there exists $\bbx\in (\fq^{\times})^n$ such that $C_1=\bbx*C_2$. Similarly, $C$ is called {\em formally self-dual} (respectively, {\em formally self-orthogonal}) if  there exists $\bbx\in (\fq^{\times})^n$ (the {\em twist}) such that $\bbx*C=C^{\perp}$ (resp. $\bbx*C\subseteq C^{\perp})$. Next we shall show a construction of this type. 

Let $C_0=(0)\subset C_1\subset\dots\subset C_n=\fq^n$  be an increasing sequence of $n+1$ linear codes in $\fq^n$,
where $C_i$ has dimension $i$ and minimum distance $d(C_i)$. This sequence is called  {\em self-dual} if  $C_i^{\perp}=C_{n-i}$ for all $i$, and  {\em formally self-dual} if there exists  $\bbx\in (\fq^{\times})^n$ such that $C_i^{\perp}=\bbx*C_{n-i}$ for all $i$. Given such a sequence, note that  for $i, j$, with $1\le i\le j\le n$ we have  
\begin{equation}
C_i^{\perp}\setminus C_j^{\perp} = \bbx*C_{n-i}\setminus \bbx*C_{n-j}=\bbx*(C_{n-i}\setminus C_{n-j})
\end{equation}
and since coordinatewise multiplication by $\bbx$ is an isometry, we have
\begin{equation}
\min\{ \wt(\bbc) \ : \ \bbc\in \bbx*(C_{n-i}\setminus C_{n-j})\} =
\min\{ \wt(\bbc) \ : \ \bbc\in (C_{n-i}\setminus C_{n-j})\} .
\end{equation}
Now, by applying Theorem \ref{qfq}, we obtain a quantum $[[n,j-i,d]]_q$ code with minimum distance
\begin{equation}
d=\min\{ \wt(\bbc) \ : \ \bbc\in(C_j\setminus C_{i})\cup(C_{n-i}\setminus C_{n-j}) \}\ge \min\{d(C_j),d(C_{n-i})\}.
\end{equation}

\subsection{Quantum codes from linear codes over $\fqd$}
\label{sec:3}

We shall denote by $\langle -,-\rangle_{H}$ the Hermitian inner product in $\mathbb{F}_{q^2}^n$, 
$\langle \bba,\bbb \rangle_{H}= \langle \bba,\bbb^q \rangle$. 
The {\em Hermitian dual} of  a linear code $C\subseteq\fqd^n$ is 
\begin{equation}
C^{\perp H}=\{ \bbv\in\fqd^n \ : \ \langle \bbv,\bbc^q \rangle=0 \mbox{ for all $\bbc\in C$} \}=(C^{q})^{\perp}.
\end{equation}
$C$ is  {\em Hermitian self-orthogonal} if $C\subseteq C^{\perp H}$, or equiva\-lently if $C^q\subseteq C^{\perp}$. Raising to the $q$-th power we find that $\langle \bbv,\bbc^q \rangle=0$ iff $\langle \bbv^q,\bbc \rangle=0$, hence $(C^{q})^{\perp}=(C^{\perp})^{q}$ and so $d(C^{\perp H})=d(C^{\perp})$.
We can derive quantum codes from classical codes over $\fqd$ as follows  \cite[Cor. 2.16]{KM}. 

\begin{theorem}\label{qfqd}
Let $C$ be a linear code over $\fqd$ of parameters $[n,k,d]$ which is self-orthogonal with respect to the Hermitian inner product. Then there exists a $[[n,n-2k,d]]_q$ quantum code with minimum distance 
$d=\min\{ \wt(\bbc) \ : \ \bbc\in(C^{\perp H}\setminus C) \}\ge d(C^{\perp})$. 
\end{theorem}

Let $C_0=(0)\subset C_1\subset\dots\subset C_n=\fqd^n$  be an increasing sequence of $n+1$ linear codes in $\fqd^n$.
From Theorem \ref{qfqd} we have the following three procedures to obtain quantum codes over $\fq$.

{\bf (A)} When the sequence is self-dual then the theorem can be applied. As $C_n=\fqd^n$, for given  $i$ let $q(i)$ be the smallest index such that $C_{i}^q\subseteq C_{q(i)}$. When $i+q(i)\le n$ we have $C_i^q\subseteq C_{q(i)}\subseteq C_{n-i} = C_i^{\perp}$, and hence we get a $[[n,n-2i,\ge d(C_{n-i})]]_q$ code.  In later sections of this article we will get sequences of one-point AG codes for which it is possible to give an estimate of $q (i)$ and consequently obtain quantum codes by this procedure. 

When the sequence $C_0=(0)\subset C_1\subset\dots\subset C_n=\fqd^n$ is formally self-dual but not self-dual, then  there exists $\bbx\in(\fq^{\times})^n$ such that  $C_i^{\perp}=\bbx*C_{n-i}$ for all  $i=0,\dots,n$. Following the aforementioned ideas of Jin and Xing, \cite{JX}, we can still derive quantum codes in some cases, by slightly modifying the codes. 

\begin{lemma}
Let $C$ be a linear code over $\mathbb{F}$ and $\bbx\in(\mathbb{F}^{\times})^n$. Then
$\bbx*C^{\perp}=(\bbx^{-1}*C)^{\perp}$.
\end{lemma}
\begin{proof}
Since both codes have equal dimension, it suffices to prove one inclusion. Let $\bbv\in\bbx*C^{\perp}$. There exists $\bba\in C^{\perp}$ such that $\bbv=\bbx*\bba$. Then $\bbv*\bbx^{-1}=\bba\in C^{\perp}$, hence $0=\langle(\bbv*\bbx^{-1}),\bbc\rangle =\langle \bbv, (\bbx^{-1}*\bbc)\rangle$ for all $\bbc\in C$, so $\bbv\in (\bbx^{-1}*C)^{\perp}$.
\end{proof}

{\bf (B)} Suppose the twist $\bbx$ verifies $\bbx\in(\mathbb{F}_q^{\times})^n$. Since the elements of $\fq$ are precisely the $(q+1)$-th powers in $\fqd$, there exists $\bby\in\fqd^n$ such that $\bby^{q+1}=\bbx$, or equivalently $\bby^{q}=\bby^{-1}*\bbx$. Consider the sequence $\bby*C_0=(0)\subset \bby*C_1\subset\dots\subset \bby*C_n=\fqd^n$. When $i+q(i)\le n$ we have
\begin{equation}
(\bby*C_i)^q=\bby^q*C^q_i=\bby^{-1}*\bbx*C^q_{i}\subseteq\bby^{-1}*\bbx*C_{q(i)}\subseteq \bby^{-1}*\bbx*C_{n-i}= (\bby*{C_i})^{\perp}
\end{equation}
and we can apply Theorem \ref{qfqd} to get a $[[n,n-2i,\ge d(C_{n-i})]]_q$ quantum code. 

{\bf (C)} To give a concrete example in which the condition  $\bbx\in(\mathbb{F}_q^{\times})^n$ holds, we can consider the case in which  there exists a self-dual sequence $C'_0=(0)\subset C'_1\subset\dots\subset C'_n=\fq^n$   of codes over $\fq$ with ${C'_i}^{\perp}=\bbx*C'_{n-i}$ and such that $C_0,C_1,\dots,C_n$ are the codes over $\fqd$ spanned by $C'_0, C'_1,\dots, C'_n$, respectively. In this situation it is clear that $C_i$ and $C'_i$ have the same parameters,  $C_i^{\perp}=\bbx*C_{n-i}$ and $({C_i})^q=C_i$. Let $\bby\in\fqd^n$ such that $\bby^{q+1}=\bbx$ and consider now the sequence   $\bby*C_0=(0)\subset \bby*C_1\subset\dots\subset \bby*C_n=\fqd^n$.  For $2i\le n$ we have
\begin{equation}
(\bby*C_i)^q=\bby^q*C_i=\bby^{-1}*\bbx*C_i\subseteq \bby^{-1}*\bbx*C_{n-i}=\bby^{-1}*{C_i}^{\perp}= (\bby*{C_i})^{\perp}
\end{equation}
so we can apply Theorem \ref{qfqd} to get a $[[n,n-2i,\ge d(C_{n-i})]]_q$ code.

\section{Algebraic Geometry codes and Castle codes} 
\label{sec:4}

As we have seen, in order to obtain a quantum code from a linear classical code $C$, we have to check the self-orthogonality (or Hermitian self-orthogonality) of $C$ and we have to compute the dual distance $d(C^{\perp})$ (resp. $d(C^{\perp H})$).  Both tasks are  difficult in general. For this reason, among all linear codes producing quantum codes, the class of Algebraic Geometry (AG) codes have received considerable attention \cite{GH,GHR,Ji,JX,KM,Sh}. For these codes there is a simple criterion of self-orthogonality. In addition, certain families of AG codes allow efficient methods to estimate their minimum distances \cite{GMRT,HLP,OM,Sti}.

Let us briefly remember the construction of AG codes (see \cite{HLP,OM} for details and proofs). We pay particular attention to AG codes coming from Castle and weak Castle curves, that produce self-dual and formally self-dual sequences of codes in a natural way.

\subsection{Algebraic Geometry codes}
\label{sec:5}

Let $\cX$ be a nonsingular, projective, geometrically irreducible algebraic curve of genus $g$ defined over $\fq$. Take two rational divisors $D$ and $G$ on $\cX$ with disjoint supports and such that $D$ is the sum of $n$ rational distinct points, $D=P_1+\dots+P_n$. In what follows we shall assume $n>2g$. The algebraic geometry code defined by the triple $(\cX,D,G)$ is  
\begin{equation}
C=C(\cX,D,G)=\{ (f(P_1),\dots,f(P_n)) : f\in \cL(G)  \}
\end{equation}
where $\cL(G)=\{ f\in \fq(\cX)^{\times} : G+\divv(f)\ge 0\} \cup \{0\}$ is the Riemann-Roch space associated to $G$. The dimension and minimum distance of this code verify
\begin{equation}\label{GoppaBound}
k=\ell(G)-\ell(G-D), \; d\ge n-\degg(G)+\gamma_{a+1}
\end{equation}
where $\ell(G)$ is the dimension of $\cL(G)$, $a=\ell(G-D)$ is the {\em abundance} of $C$ and for $r\ge 1$,  
$\gamma_r=\min \{ \degg(A) : \mbox{$A$ is a rational divisor on $\cX$ with $\ell(A)\ge r$}\}$ is the $r$-th gonality of $\cX$. 

From the residue theorem, it follows that the dual of an AG code is again an AG code. More precisely, $C(\cX,D,G)^{\perp}=C(\cX,D,D+W-G)$, where $W=\divv(\omega)$ is the divisor of a differential form $\omega$ with simple poles and residue 1 at every point $P_i$ in the support of $D$. Thus $C$ is self orthogonal if $G\le D+W-G$. When $C$ is defined over the field $\fqd$,  and since $C(\cX,D,G)^q\subseteq C(\cX,D,qG)$, $C$ is Hermitian self-orthogonal if $qG\le D+W-G$. More generally, if $\eta$ is a differential form  with simple poles at every point $P_i$ in  $\supp(D)$, then $C(\cX,D,D+\divv(\eta)-G)=\bbx*C(\cX,D,G)^{\perp}$, where $x_i\neq 0$ is the residue of $\eta$ at $P_i$.

\subsection{One-point Algebraic Geometry codes}
\label{sec:6}

If  $Q\in\cX(\fq)$ is a rational point not in the support of $D$, we can take $G=mQ$. The obtained code is called {\em one-point}. These are the most known and studied codes among the whole family of AG codes. 
Let $\cL(\infty Q)=\cup_{r\ge 0}\cL(rQ)$ be the space of rational functions having poles only at $Q$ and consider the evaluation map $\ev: \cL(\infty Q)\rightarrow\fq^n$, $\ev(f)=(f(P_1),\dots,f(P_n))$.
The parameters of  $C(\cX,D,mQ)=\ev(\cL(mQ))$  are closely related to the {\em Weierstrass semigroup} of $Q$
\begin{equation}
S=S(Q)=\{ -v_Q(f) : f\in\cL(\infty Q)\}=\{ 0=\rho_1<\rho_2<\cdots\}
\end{equation}
where $v_Q$ is the valuation at $Q$.
In what follows, for simplicity we shall write $v(f)=-v_Q(f)$. If for all $i=0,1,2,\dots$ we take a function $f_i\in\cL(\infty Q)$ such that $v(f_i)=\rho_i$, then $\{f_1,f_2,\dots\}$ is a basis of $\cL(\infty Q)$. From the set of codes $\{ C(\cX,D,mQ) : m=0,\dots,n+2g\}$ we can extract  an increasing sequence $C_0=(0)\subset C_1\subset\dots\subset C_n=\fq^n$.  This sequence is self-dual if the differential form $\omega$  with simple poles and residue 1 at every point $P_i$ in  $\supp(D)$ has zeros only at $Q$.
 This was first noted by Stichtenoth \cite{SD}. The sequence is formally self-dual if the divisors
$D+W$ and $(n+2g-2)Q$ are equivalent, that is when $(n+2g-2)Q-D$ is a canonical divisor,  $(n+2g-2)Q-D\sim\divv(\eta)$, and then the twist is $\bbx=(\res_D(\eta))$.

\subsection{Castle curves and codes}
\label{sec:7}

Castle curves and codes were introduced in \cite{Castle2}, where it is shown that many of the most interesting known examples of AG codes belong to this family. Let $\cX$ be a curve defined over $\fq$ and $Q$ a rational point. The Lewittes bound \cite{OM} implies that $\#\cX(\fq)\le q\rho_2+1$, where $\rho_2$ is the {\em multiplicity} of $S(Q)$, that is the first nonzero element of $S(Q)$. The pointed curve $(\cX,Q)$ is called {\em Castle} if 
\begin{enumerate}
\item [\rm (C1)] $S(Q)$ is symmetric; and
\item [\rm (C2)] $\#\cX(\fq)=q\rho_2+1$.
\end{enumerate}
Thus rational curves, Hermitian curves, Suzuki and Ree curves, etc., verify the  Castle conditions. 
This family can be generalized in the following way:  a pointed curve $({\mathcal  X}, Q)$ over ${\mathbb F}_q$ is called {\em weak Castle} if it satisfies the conditions
\begin{enumerate}
\item [\rm (C1)] $S(Q)$ is symmetric; and
\item [${\rm (C2')}$] let ${\mathbb P}^1$ be the projective line over the algebraic closure of $\fq$. There exist a morphism $f:{\mathcal  X}\to{\mathbb P}^1$ with
$\mbox{div}_\infty(f)=\ell Q$, and a set $U=\{\alpha_1,\ldots,\alpha_h\}\subseteq {\mathbb F}_q$ such that for all $i=1,\ldots,h$, we have $f^{-1}(\alpha_i)\subseteq {\mathcal  X}({\mathbb F}_q)$ and  $\# f^{-1}(\alpha_i)=\ell$.
\end{enumerate}

Castle curves are weak Castle: simply take $f\in\cL(\infty Q)$ with $v(f)=\rho_2$ and $U=\fq$  \cite{Castle2}. Conversely, in the above situation of {\rm(C2')}, observe that $\ell\in S(Q)$. Furthermore, since $f$ is unramified over each $\alpha_i$, if we write $f^{-1}(\alpha_i)=\{P_1^i,\ldots,P_{\ell}^i\}$, then $\mbox{div}(f-\alpha_i)=\sum_{j=1}^{\ell} P_j^i-\ell Q$. Let 
\begin{equation}\label{divD}
D=D_{U,f}=\sum_{i=1}^h\sum_{j=1}^{\ell}P_j^i .
\end{equation}
If $(\cX,Q)$ is weak Castle and $D$ is the sum of all rational points different from $Q$, then it is said to be {\em complete}.
The one-point codes  $C(\cX,D,mQ)$ of length $n=\ell h$  are called {\em weak Castle codes}, or simply {\em Castle codes} if  $(\cX,Q)$ satisfies ${\rm (C2)}$. 

\begin{example}\label{hiperelipticosimpar}
Let $q$ be  odd. A hyperelliptic curve $\cX$ over $\fq$ is given by an equation $y^2=F(x)$, where $F$ is a squareless polynomial. If $\deg(F)=2g+1$ then $\cX$ has genus $g$ and one hyperelliptic point at infinity, $Q$. Then  $\cX$ is Castle iff $F(\alpha)$ is a nonzero square for all $\alpha\in\fq$. For example, the curve $y^2=x^q-x+1$  has $2q+1$ points and it is Castle for all $q$.
Otherwise, if $\cX$ is not Castle, take $U=\{ \alpha\in\fq : F(\alpha) \mbox{ is a nonzero square in $\fq$}  \}$. Whenever $U\neq\emptyset$, $(\cX,Q)$ is a weak Castle curve and provides  codes of length $n=2\# U$ . Similarly,
if $q$ is  even, a hyperelliptic curve $\cX$ of genus $g$ over $\fq$ is given by an equation $y^2+y=F(x)$, where $F(x)$ is a rational function with $\deg(\divv_0(F)),\deg(\divv_{\infty}(F))\le 2g+2$. If $F$ is a polynomial there is one hyperelliptic point at infinity $Q$ and $\cX$ is Castle or complete weak Castle. For example,  the curves $y^2+y=x^{u}$,  where $q+1|u$ and $0<u\le q^2-1$ are Castle over $\fqd$. 
\end{example}

Castle and weak Castle curves have several nice properties that can be translated to codes arising from them. Next we describe some ones which are relevant for our study. Proofs can be found in \cite{Castle2}.

\begin{proposition}\label{weakCsequence}
Let $(\cX,Q)$ be a weak Castle curve of genus $g$ over $\fq$ and $C(\cX,D,mQ)$ a weak Castle code. The following hold. 
\begin{enumerate}
\item[(1)] The divisors $D$ and $nQ$ are equivalent. Then
for $m<n$,  $C(\cX,D,mQ)$ reaches the Goppa bound if and only if $C(\cX,D,(n-m)Q)$ does.
For $m\ge n$, $C(\cX,D,mQ)$ is an abundant code of abundance $l((m-n)Q)$.
\item[(2)] $(2g-2)Q$ is a canonical divisor. In consequence $(n+2g-2)Q-D$ is also a canonical divisor and there exists $\bbx\in (\fq^{\times})^n$ such that $C(\cX,D,mQ)^{\perp}=\bbx*C(\cX,D,m^{\perp}Q)$, where $m^{\perp}=n+2g-2-m$. 
\item[(3)] For $i=1,\dots,n$, let $m_i=\min\{ m : \ell(mQ)-\ell((m-n)Q)\ge i\}$. Then $C_i=C(\cX,D,m_iQ)$ has dimension $i$ and the sequence   $C_0=(0)\subset C_1\subset\dots\subset C_n=\fq^n$  is a formally self-dual sequence of codes.
\end{enumerate}
\end{proposition}

The twist $\bf x$ of item (2) does not depend on $m$ and can be computed as showed in Section \ref{sec:6}.
It is also interesting to note that the set $M=\{ m_1=0,m_2,\dots,m_n\}$,  is called the {\em dimension set} of $(\cX,Q)$ and  can be used to obtain good estimates on the minimum distance of codes $C(\cX,D,m_iQ)$  by applying the order bound \cite{GMRT,OM}. When
$\cX$ is weak Castle this set is easy to compute as  $M=S(Q)\setminus (n+S(Q))$.

Thus, from a weak Castle curve we obtain in a natural manner a formally self-dual sequence of codes and consequently, by applying the procedures explained in Sections \ref{sec:2} and \ref{sec:3}, a set of quantum stabilizer codes. To be precise, keeping the notation 
of Proposition \ref{weakCsequence}, we have the following results.

\begin{corollary}
Let $(\cX,Q)$ be a weak Castle curve over $\fq$ and let  $C_0=(0)\subset C_1\subset\dots\subset C_n=\fq^n$ be the sequence of codes obtained from it by  Proposition \ref{weakCsequence}(3). If $2i \le n$ then we have a quantum code of parameters $[[n,n-2i,\ge d(C_{n-i})]]_q$, with $d(C_{n-i})\ge m-m_{n-i}+\gamma_{a+1}$, where $a=\ell((m_{n-i}-n)Q)$ is the abundance of $C_{n-i}$.
\end{corollary}
\begin{proof}
Apply construction (C) of Section \ref{sec:3} to the formally self-dual sequence $C_0\subset\dots\subset C_n$. The estimate on the minimum distance follows from the improved Goppa bound stated in equation (\ref{GoppaBound}).  
\end{proof}

\begin{corollary}
Let $(\cX,Q)$ be a weak Castle curve over $\fqd$ such that the sequence  $C_0=(0)\subset C_1\subset\dots\subset C_n=\fqd^n$ obtained  by  Proposition \ref{weakCsequence}(3) is self-dual.  If $qm_i\le m_{n-i}$ then we have a quantum code over $\fq$ of parameters $[[n,n-2i,\ge d(C_{n-i})]]_q$, with $d(C_{n-i})\ge m-m_{n-i}+\gamma_{a+1}$, where $a=\ell((m_{n-i}-n)Q)$. 
\end{corollary}
\begin{proof}
Similar to the previous corollary, now applying construction (A) and taking into account that
$C(\cX,D,m_iQ)^q\subseteq C(\cX,D,qm_iQ)$.
\end{proof}

\subsection{Duality of Castle codes}
\label{sec:8}

In order to obtain quantum codes using the procedure (A) of Section \ref{sec:3}, formal duality is not enough.   In this subsection we give 
some sufficient conditions on curves verifying Castle assumptions in order to obtain codes satisfying the self-duality property. 
Given a pointed curve  $(\cX,Q)$  of genus $g$, satisfying the condition ${\rm (C2')}$, we consider  the divisor $D$ of equation (\ref{divD}) and the intrinsic rational function $\phi:=\prod_{i=1}^h(f-\alpha_i)$ related to $f$. As a notation, given an integer $m$ we write $m^{\perp}=n+2g-2-m$.

\begin{lemma} \label{dz} 
Let $(\cX,Q)$ be a pointed curve satisfying ${\rm (C2')}$. If $\divv(d\phi)=(2g-2)Q$, then $C(\cX, D,mQ)^\perp=C(\cX,D,m^{\perp}Q)$.
\end{lemma}
\begin{proof}
Consider the differential form $\eta=d\phi/\phi$. From hypothesis we get $\divv(\eta)=(n+2g-2)Q-D$, and thus $\eta$ has simple poles and residue $1$ at $P\in \mbox{supp}(D)$. 
\end{proof}

\begin{proposition} \label{df}
Let $(\cX,Q)$ be a pointed curve  satisfying {\rm(C2)} and let $f\in\cL(\infty Q)$ such that $v(f)=\rho_2$. If $\divv(df)=(2g-2)Q$, then $C(\cX,D,mQ)^\perp=C(\cX,D,m^{\perp}Q)$.
\end{proposition}
\begin{proof}
Since in this case $\phi=f^q-f$, we obtain $\divv(d\phi)=\divv(df)$ and thus the result follows from the previous lemma. 
\end{proof}

\begin{corollary}\label{cor} 
Let $(\cX,Q)$ be a pointed curve in the conditions of Lemma $\ref{dz}$ or Proposition $\ref{df}$. Then $C(\cX,D,mQ)$ is Euclidean self-orthogonal if $2m\leq  n+2g-2$. If $q$ is a square, $C(\cX,D,mQ)$ is Hermitian self-orthogonal if $(\sqrt{q}+1)m\leq n+2g-2$.
\end{corollary}
 
\begin{corollary}\label{cor2} 
Let $(\cX,Q)$ be a pointed curve over $\fqd$ in the conditions of Lemma $\ref{dz}$ or Proposition $\ref{df}$.
Let $M$ be the dimension set of $(\cX,Q)$.   If $(q+1)m_i\leq n+2g-2$ then we have a quantum code over $\fq$ of parameters $[[n,n-2i,\ge d(C_{n-i})]]_q$, with $d(C_{n-i})\ge m-m_{n-i}+\gamma_{a+1}$, where $a=\ell((m_{n-i}-n)Q)$. 
\end{corollary}

\begin{remark}\label{remarkhyper}
Note that the conditions required by  Proposition \ref{df} are not verified by all Castle curves. Consider, for example, the hyperelliptic curve $\cX$ of equation $y^2=x^q-x+1$ over $\fq$, $q$ odd of Example \ref{hiperelipticosimpar}. Here $\rho_2=2$ which  is the pole order of $x$ at $Q$, and $\cX$ is Castle with $f=x$ in the notation of Proposition \ref{df}. As the points $P$ over the $x=\alpha$, for $\alpha$ a root of $x^q-x+1$ in $\overline{\mathbb{F}}_q$, are also ramified, we have  $\divv(dx)\neq (2g-2)Q$.  On the contrary, for the curves $y^2+y=x^{u}$,  where $q+1|u$,  over $\fqd$ of Example \ref{hiperelipticosimpar}, a simple computation shows that $\divv(dx)= (2g-2)Q$, hence they provide self-dual sequences of codes.
\end{remark}

\subsection{Curves defined by separate equations}
\label{sec:9}

Next we consider the family of curves having a plane model given by a separated variable equation $F(y)=G(x)$, where $F$ and $G$ are univariate polynomials of coprime degrees. The particular case in which one of the polynomials is linearized is interesting since it contains many of the most  relevant curves for Coding Theory purposes. For instance, several of the curves studied below were already treated in \cite{SD}.
Remember that a polynomial $F$ is called {\em linearized} (or a $q$-polynomial) if  the exponents of all monomials are powers of $q$. 

\begin{proposition}\label{sep}
Let $\cX$ be the curve defined over $\fq$ by the equation  $F(y)=G(x)$, where $F$ and $G$ are polynomials of degrees $a=\deg(F), b=\deg(G)$ with $\gcd(a,b)=1$. Let $Q$ be the common pole of $x$ and $y$ and $\rho_2$ the multiplicity of $S(Q)$. Then 
\begin{enumerate}
\item[(1)] The genus of $\cX$ is $g=(a-1)(b-1)/2$.
\item[(2)] The Weierstrass semigroup of $Q$ is $S(Q)=\langle a,b\rangle$. In particular $S(Q)$ is symmetric and $\rho_2=\min\{a,b\}$.
\item[(3)] If $F$ is a linearized separable polynomial, then $\divv(dx)=(2g-2)Q$. If $b>a$ then $\rho_2=a$ and there exists $f:\cX\to \mathbb{P}^1$ with $\divv_\infty(f)=\rho Q$ and $\divv(df)=(2g-2)Q$.
\end{enumerate}
\end{proposition}
\begin{proof} 
(1) It follows  from the Riemann-Hurwitz-Zeuthen  genus formula. (2) From the equation of $\cX$ we have $av(y)=bv(x)$ and since $(a,b)=1$ we conclude that $v(y)=b, v(x)=a$, hence $S(Q)\supseteq\langle a,b\rangle$. Since the genus of $\langle a,b\rangle$ is $g$ we get equality. As all semigroups generated by two elements are symmetric, the  statement follows. (3)  Since $F$ is linearized, its derivative is a constant what implies that $Q$ is the only point of $\cX$ which ramifies over the morphism $x$. Using $\divv(dx)=-2\divv_{\infty}(x)+R_x$, where $R_x$ is the ramification divisor of the morphism $x$, we obtain the equality between the divisors. To see the last statement it is enough to take $f=x$.
\end{proof}

Proposition \ref{sep} gives a large family of curves verifying the conditions of Proposition \ref{df}, and hence providing self-dual sequences of Castle and weak Castle codes. Concrete examples of relevant Castle curves used in Coding Theory  belonging to this family include  those already mentioned above: Hermitian curve, Norm-Trace curve, etc.  Recall, however, that not all Castle curves satisfying the conditions of Proposition \ref{df} belong to this family. For instance, this happens to the Ree curve  \cite{Castle2}.

Next we show some examples of curves in the conditions of Lemma \ref{dz}.

\begin{example} \label{normtrace}
Let $\cX$ be the curve given by 
\begin{equation}
y^{q^{r-1}}+y^{q^{r-2}}+\cdots+y^q+y=x^u
\end{equation} 
defined over $\mathbb{F}_{q^r}$ with $u\mid (q^r-1)/(q-1)$. We shall see that it is a weak Castle curve with $q^{r-1}(u(q-1)+1)+1$ rational points and genus $g=(u-1)(q^{r-1}-1)/2$. Note that this curve is covered by the Norm-Trace curve over $\mathbb{F}_{q^r}$; indeed, it is a quotient curve of the Norm-Trace curve by the cyclic group of automorphisms generated by $(x,y)\mapsto (\zeta x,y)$, where $\zeta$ is a primitive $u$-th root of unity. Also this equation includes many other interesting curves as the Hermitian curve and its quotients in case of $r=2$. Now, consider the multiplicative subgroup of $\mathbb{F}_{q^r}^\times$
\begin{equation}
U^{\times}=\{\beta\in\mathbb{F}_{q^r}^\times \ ; \ \beta^u\in \fq\}
\end{equation}
of cardinality $u(q-1)$ and set $U=U^{\times}\cup\{0\}$. Observe that the elements of $U$ are the elements of $\mathbb{F}_{q^r}$ which split totally over the morphism $x:\cX\to \mathbb{P}^1$. Now take 
\begin{equation}
\phi=\prod_{\beta\in U} (x-\beta)=x\prod_{\beta\in U^{\times}}(x-\beta)=x(x^{u(q-1)}-1)=x^{u(q-1)+1}-x.
\end{equation} 
Since from Proposition \ref{sep} we have $\divv(dx)=(2g-2)Q$, where $Q$ is the unique pole of $x$, a sufficient condition to have $\divv(d\phi)=(2g-2)Q$ is that $u\equiv 1 \ (\mbox{mod} \ p)$, where $p$ is the characteristic of $\mathbb{F}_{q}$. Notice that this always happens in characteristic $2$.
\end{example}

\subsection{Maximal Curves}
\label{sec:10}

Recall that a curve $\cX$ of genus $g$ over a field $\fqdr$  is maximal if it attains equality in the Hasse-Weil bound $\# \cX(\fqdr)=q^{2r}+1+2gq^r$. Maximal curves often provide AG codes with good parameters.  In this subsection we shall give examples of maximal Castle curves verifying the conditions of Proposition \ref{sep}. In particular, the examples we present are  Artin-Schereier curves defined over $\fqdr$ by plane models of type $F(y)=G(x)$, where $F(y)$ is linearized and separable  of degree $\rho$ a power of $p=\mbox{char}(\fq)$, and $G(x)$ of degree $q^{r}+1$, bigger than $\rho$. Thus these curves have genus $2g=(\rho-1)q^r$ and  the Weierstrass semigroup $S(Q)$ of $Q$ is generated by $\rho$ and $q^{r}+1$. In this situation, maximality implies the Castle condition $\# \cX(\fqdr)=q^{2r}\rho+1$. We just show a couple of examples.

\begin{example}\label{ex3} 
Let $a\in\fqd$ such that $a^q+a=0$ and consider the curve $\cX$ defined over $\fqd$ by
\begin{equation}
y^{q/p}+y^{q/p^2}+\ldots+y=ax^{q+1}
\end{equation}
where $p=\car(\fq)$. To check the maximality of $\cX$ we note that it is covered by the Hermitian curve
 $\cH: w^q+w=z^{q+1}$ via the morphism  $x=z$, $y=(-aw)^p+aw$. The cardinality of $\cX(\fqd)$ can be computed from this fact by taking into account that its genus  $g$ satisfies $2g=(q/p-1)q$.
\end{example}
  
\begin{example}\label{ex2} 
Let $\cX$ be the curve over $\fqdr$, $r$ odd, defined by
\begin{equation}
y^q+y=x^{q^r+1} .
\end{equation}
As in the previous example, $\cX$ is maximal as it is covered by the Hermitian curve $\cH: w^{q^r}+w=z^{q^r+1}$ over $\fqdr$. A covering is given by $x=z$, $y=w^{q^{r-1}}-w^{q^{r-2}}+\dots+w$. 
\end{example}

\section{Some Examples of quantum codes}
\label{sec:11}

In this Section we shall show a few examples of quantum codes  obtained from the curves treated in the above Section.  Computations have been done by using the computer system  Magma \cite{magma}. In order to compare the quality of the obtained parameters we shall use the tables \cite{tablas,tablas2} and the quantum
Gilbert-Varshamov bound \cite{GVB}:

\begin{theorem}
Suppose that $n>k\ge 2$, $d\ge 2$ and $n\equiv k \mbox{\rm (mod $2$)}$. Then there exists a stabilizer quantum code $[[n,k,d]]_q$ provided that
\begin{equation}
\frac{q^{n-k+2}-1}{q^2-1}>\sum_{i=1}^{d-1}(q^2-1)^{i-1} {{n}\choose{i}}.
\end{equation}
\end{theorem}

We shall use the following notation: given a code $[[n,k,d]]_q$, we write $[[n,k,d]]_q^{\dag}$ when the parameters $n,k,d$ meet the Gilbert-Varshamov bound with equality. We  write $[[n,k,d]]_q^{\ddag}$ when the parameters $n,k,d$ strictly improve on the Gilbert-Varshamov bound.

\subsection{Codes from the Suzuki curve}
\label{sec:12}

The curve $\cX$ of equation $y^q+y=x^{q_0}(x^q+x)$ over $\fq$, where $q_0=2^s$ and $q=2q_0^2$ is called the {\em Suzuki} curve. It has genus $g=q_0(q-1)$ and $q^2+1$ rational points. Let $Q$ be the common pole of $x$ and $y$. Then the Weierstrass semigroup $S(Q)$ is generated by $\langle  q,q+1,q+2q_0, q+2q_0+1\rangle$ (see \cite{OM} and the references therein), so $\cX$ is a Castle curve. Codes arising from this curve have been extensively studied and their parameters are rather well known. In particular, the conditions stated in Proposition \ref{df} and Corollary \ref{cor} holds, so that the duals of one point Suzuki codes are again one-point Suzuki codes. We can use them to obtain quantum codes from construction (C) of Section \ref{sec:3}.

\begin{example}
For $q=8$, $C(\cX,D,mQ)$ has length $64$ and is self-orthogonal whenever $m\le 45$. Then we get quantum codes over $\mathbb{F}_8$ of parameters 
$[[64,62,2]]_8^{\dag}$, 
$[[64,54,3]]_8$, 
$[[64,52,4]]_8^{\dag}$, 
$[[64,42,5]]_8$, 
$[[64,40,6]]_8$,
$[[64,38,7]]_8$, 
$[[64,36,8]]_8$.
 These are good parameters compared with those listed in \cite{tablas}. 
\end{example}

\subsection{Codes from elliptic and hyperelliptic curves}
\label{sec:13}

\begin{example}\label{Jhyp}
Consider the family of hyperelliptic curves $y^2+y=x^{u}$,  where $q+1|u$ and $u\le q^2-1$  over $\fqd$, $q$ even, of Example \ref{hiperelipticosimpar}. In Remark \ref{remarkhyper} we noted that they are Castle and  verify $\divv(dx)= (2g-2)Q$. Since they have genus $g=\lfloor u-1\rfloor/2$,  $C(\cX,D,mQ)$  is self-orthogonal whenever $2m\le 2q^2+u-3$ and Hermitian self-orthogonal whenever $(q+1)m\le 2q^2+u-3$.   We find quantum codes with  parameters
$[[8,6,2]]_2^{\ddag}$,
$[[32,30,2]]_4^{\ddag}$,
$[[32,24,4]]_4^{\ddag}$,
$[[128,126,2]]_8^{\ddag}$,
$[[128,116,4]]_8^{\dag}$,
$[[128,112,6]]_8^{\ddag}$,
$[[128,108,8]]_8^{\ddag}$
among others.
The particular case $u=q+1$ was studied in \cite[Sec. III (A)]{Ji}.
\end{example}

\begin{example}
Here we show an example of a complete weak Castle curve which is not Castle. When $q\equiv 3$ (mod 4) then $-1$ is not a square, hence the elliptic curve $\cX:y^2=x^3+ax$ has $q+1$ rational points over $\fq$. Consequently it has $q^r+1+2\sqrt{q^r}$ rational points over $\fqr$ for $r\equiv 2$ (mod 3).  To give a concrete example take $q=3, r=2$. A simple computation shows that $\rho_2=2$ and $\cX$ is not Castle. However $\cX$ is a complete weak Castle curve with respect the morphism $f=y$.  We obtain codes with the following parameters over ${\mathbb F}_9$:
$[[15,13,2]]_9^{\dag}$, 
$[[15,7,4]]_9^{\dag}$, 
$[[15,5,5]]_9^{\dag}$, 
$[[15,3,6]]_9^{\dag}$, 
$[[15,1,7]]_9$.
Quantum codes from elliptic curves have been treated in \cite{JX}.
\end{example}

\subsection{Codes from Hermitian, Norm-Trace and related curves} 
\label{sec:14}

Quantum codes arising from Hermitian curves have been extensively treated, \cite{KM,SK}. Here we shall consider the related more general curves $\cX$ of Example \ref{normtrace} given by an equation of the form
\begin{equation}
y^{q^{r-1}}+y^{q^{r-2}}+\cdots+y^q+y=x^u
\end{equation}  
over $\mathbb{F}_{q^r}$ with $u\mid (q^r-1)/(q-1)$. We saw that when $u\equiv 1 \ (\mbox{mod} \ p)$ then $\cX$ satisfies the conditions of Lemma \ref{dz}, where $p$ is the characteristic of $\mathbb{F}_{q^r}$. Note that when $u=(q^r-1)/(q-1)$, then $\cX$ is the Norm-Trace curve, and if $r=2$ then $\cX$ is the Hermitian curve. In both cases  such conditions are also automatically satisfied.

\begin{example} 
For $q=2$, $r=4$ and $u=3$, from  Corollary \ref{cor} the codes $C(\cX, D, mQ)$ of length $32$ over $\mathbb{F}_{16}$ are Hermitian self-orthogonal whenever $m\leq 8$. The quantum codes over $\mathbb{F}_{4}$ obtained from these  have parameters $[[32,30,2]]_4^{\ddag}$  and 
$[[32,24,3]]_4^{\dag}$. 
\end{example}

\begin{example} 
For $q=2$, $r=3$ and $u=7$, we obtain the Norm-Trace curve. The codes $C(\cX, D, mQ)$ of length $32$  over the field $\mathbb{F}_8$ are self-orthogonal if $m\leq 24$. We  obtain 
$[[32,28,2]]_8^{\dag}$,
$[[32,26,3]]_8^{\ddag}$,
$[[32,18,4]]_8$
quantum codes. 
\end{example}

\begin{example} 
For $r=2$ we obtain the Hermitian curves and  quotients $y^q+y=x^u$ with $u\mid q+1$, over $\mathbb{F}_{q^2}$, which are $\mathbb{F}_{q^2}$-maximal. The case where $q$ is an odd power of 2 and $u=3$ was considered in \cite[Sec. III (B)]{Ji}. We obtain codes
$[[8,6,2]]_2^{\ddag}$,
$[[64,54,3]]_4^{\dag}$,
$[[64,52,4]]_4^{\dag}$,
$[[176,162,3]]_8$,
$[[176,156,5]]_8$,
$[[176,154,6]]_8$,
$[[176,150,8]]_8^{\dag}$,
$[[176,146,9]]_8^{\dag}$.
\end{example}

\subsection{Codes from maximal curves} 
\label{sec:15}

To end this section, let us see some quantum codes arising from maximal curves as the ones considered in Section \ref{sec:10}.

\begin{example}
Let $\cX$ be the curve over $\fqd$ of Example \ref{ex3} 
\begin{equation}
y^{q/p}+y^{q/p^2}+\ldots+y=ax^{q+1}
\end{equation}
where $a\in \fqd$ verifies $a^q+a=0$ and $p$ is the characteristic of $\fq$. For $q=9$ let us take the curve  $y^3+y=\alpha^5x^{10}$ over $\mathbb{F}_{81}$, where $\alpha$ is a primitive element of $\mathbb{F}_{81}$. It has $244$ rational points and genus $9$. The codes $C(\cX,D,mQ)$ are Hermitian self-orthogonal if $m\leq 25$, by Corollary \ref{cor}. For these values we obtain quantum codes with the following parameters: 
$[[ 243, 241, 2 ]]_9^{\ddag}$, 
$[[ 243, 233, 3 ]]_9^{\dag}$, 
$[[ 243, 219, 6 ]]_9$ and 
$[ [ 243, 213, 9 ]]_9^{\dag}$. 
For $q=8$ take $a=1$. The curve has 257 rational points and genus 12. From Corollary \ref{cor}, the codes $C(\cX,D,mQ)$ are Hermitian self-orthogonal if $m\leq 30$. With these values we can provide quantum codes with the following parameters: 
$[[ 256, 254, 2 ]]_8^{\ddag}$,
 $[[ 256, 248, 3 ]]_8^{\dag}$, 
 $[[ 256, 238, 4 ]]_8$ and 
 $[[ 256, 224, 8 ]]_8$. 
 These quantum codes have good parameters compared to similar quantum codes in \cite{tablas}.
\end{example}

\begin{example} 
This is a  particular case of the curve given in Example \ref{ex2}. Take $r=3$ in that example and
let  $\cX$ be the curve over $\mathbb{F}_{q^6}$ given by the equation
\begin{equation}
y^q+y=x^{q^3+1}.
\end{equation}
When $q=2$,  $\cX$ has genus 4 and 129 rational points. From Corollary \ref{cor}, the codes $C(\cX,D,mQ)$ are Hermitian self-orthogonal if $m\leq 14$. With these values we obtain the following quantum codes: 
$[[ 128, 126, 2 ]]_8^{\ddag}$, 
$[[ 128, 116, 4 ]]_8^{\dag}$,
$[[ 128, 112, 6 ]]_8^{\ddag}$ and 
$[[ 128, 108, 8 ]]_8^{\ddag}$.
\end{example}

\section{Traces of Castle codes}
\label{sec:16}

Some of the codes we have obtained in the previous sections are defined over large extensions of $\fq$. However, in practice we can be interested in codes defined over smaller fields. Such codes can be obtained from the formers by two ways:  given a code $C$ over $\fqr$, we can consider the subfield subcode of $C$ over $\fq$,  $C|\fq=C\cap\fq^n$, and the trace code $\tr(C)$ (see \cite{Sti}, Chapter VIII, where a detailed treatment of subfield subcodes is done). Both codes over $\fq$ are closely related by Delsarte's theorem,  $(C|\fq)^{\perp}=\tr(C^{\perp})$. However, as we shall see, at least for AG codes the way of using trace maps is simpler.
Remark however that  quantum codes from subfield subcodes of affine variety codes have been studied in \cite{GH,GHR}.

As in the previous sections, let $(\cX,Q)$ be a pointed curve, now defined over $\fqr$.
Let $\tr:\fqr\rightarrow\fq$ be  the trace map, defined as $\tr(x)=x+x^q+\cdots+x^{q^{r-1}}$. $\tr$ is a surjective $\fq$-linear map
that can be extended to $\fqr^n$ (coordinatewise) and to $\cL(\infty Q)$. Here we list some properties of these extensions. 

\begin{lemma}\label{trproperties}
The following properties hold.
\begin{enumerate}
\item[\rm (a)]  Both extensions of $\tr$ are $\fq$-linear and $\tr:\fqr^n\rightarrow\fq^n$ is surjective. 
\item[\rm (b)] If $f\in\cL(\infty Q)$ then $v(\tr(f))=q^{r-1}v(f)$. Thus, if $f$ is nonconstant then $\tr(f)\neq 0$.
\item[\rm (c)] For all $f\in\cL(\infty Q)$ we have $\tr(\ev(f))=\ev(\tr(f))=\ev(\tr(f^q))$. 
\end{enumerate}
\end{lemma}

Let $\cV$ be the set of functions in $\cL(\infty Q)$ evaluating to $\fq^n$, $\cV=\{ f\in\cL(\infty Q) : \ev(f)\in\fq^n\}$. $\cV$ is a $\fq$-linear subspace of $\cL(\infty Q)$ containing $\ker(\ev)$. As a consequence of Lemma \ref{trproperties}(c), $\cV$ also contains $\tr(\cL(\infty Q))$. The next proposition explains why trace codes are simpler to handle than subfield subcodes.

\begin{proposition}
$\cV=\tr(\cL(\infty Q))+\ker(\ev)$.
\end{proposition}
\begin{proof}
Since $\tr\circ\ev$ is surjective (as both maps are so), if $f\in\cV$ there exists $g\in\cL(\infty Q)$ such that $\ev(f)=\tr(\ev(g))=\ev(\tr(g))$, hence $f-\tr(g)\in\ker (\ev)$. The converse is clear.
\end{proof}

Thus $\ev(\cV)=\ev(\tr(\cL(\infty Q))$ and henceforth we will only consider trace codes. However note that for Castle codes we can also give a basis of $\ker(\ev)$ over $\fqr$. For all $m>0$ we have $\ker(\ev)\cap\cL(mQ)=\cL(mQ-D)$. Let $\varphi=f_2^{q^r}-f_2$. Since $\cX$ is Castle, we have $\divv(\varphi)=D-nQ$. Consequently there is an isomorphism $\cL((n-m)Q)\rightarrow\cL(mQ-D)$ given by $f\mapsto\varphi f$ and the set $\{ \varphi f_1,\varphi f_2,\dots\}$ is a basis of $\ker(\ev)$.

The following results allows us to obtain self-orthogonal trace codes over $\fq$ when the pointed  curve $(\cX,Q)$ of genus $g$ over $\fqr$ has the self-dual property.  Let $f_1=1$ and for $i=2,3,\dots$ let $f_i\in\cL(\infty Q)$ such that $v(f_i)=\rho_i$. In view of Lemma \ref{trproperties}(c),  if $\rho_i=q\rho_t$ then we take $f_i=f_t^q$. Then  $L=\{f_1,f_2,\dots\}$ is a basis of  $\cL(\infty Q)$ over $\fqr$ verifying $L^q\subset L$, and the sets $L_m=\{f_i\in L : \rho_i\le m\}$  are bases of $\cL(mQ)$ for all $m\ge 0$.
Let $\alpha$ be a primitive element of $\fqr$ over $\fq$. Then $\{1,\alpha,\dots,\alpha^{r-1}\}$ is a basis of $\fqr$ over $\fq$ and every element $a\in\fqr$ can be written as $a=\lambda_1+\lambda_2\alpha+\cdots+\lambda_r\alpha^{r-1}$ with $\lambda_1,\dots,\lambda_r\in\fq$.  

\begin{proposition}\label{basistr}
The set \ $\cB=\{1\}\cup\{\tr(\alpha^jf_i) : j=0,\dots,r-1, f_i\in L \mbox{ for } i>1\}$ is a basis of $\tr(\cL(\infty Q))$ over $\fq$.
The set \ $\cB_m=\{1\}\cup\{\tr(\alpha^jf_i) : j=0,\dots,r-1, f_i\in L_m \mbox{ for } i>1\}$ is a basis of $\tr(\cL(m Q))$ over $\fq$.  
\end{proposition}
\begin{proof}
It suffices to show the first statement. That $\cB$ is a generator set follows from the $\fq$-linearity of $\tr$ and the fact that $L$ is a basis of $\cL(\infty Q)$ over $\fqr$. To see that $\cB$ is an independent set, and since $v(\tr(f_i))\neq v(\tr(f_t))$ if $i\neq t$, it suffices to show that $\tr(f_i),\tr(\alpha f_i),\dots,\tr(\alpha^{r-1}f_i)$ are independent for all $i>1$. If $0=\sum_j\lambda_j\tr(\alpha^{j-1}f_i) = \tr(\sum_j\lambda_j \alpha^{j-1}f_i)$ then $\tr(af_i)=0$, where $a=\sum_j \lambda_j \alpha^{j-1}$. By Lemma \ref{trproperties}(b), we have $af_i=0$ hence $a=0$ and $\lambda_j=0$ for all $j$.
\end{proof}

For $j=0,\dots,r-1$, let   $\beta_j\in\fqr$ such that $\beta_j^q=\alpha^j$. Then, if $f_i=f_t^q$ we have $\ev(\tr(\alpha^jf_i))=\ev(\tr((\beta_jf_t)^q))=\ev(\tr(\beta_jf_t))$ by Lemma \ref{trproperties}(c).
Consider the sets $L'=\{1\}\cup (L\setminus L^q)$ and $L'_m=L'\cap L_m$. Define accordingly the sets $\cB'$ and $\cB'_m$ by substituting $L$ by $L'$ and $L_m$ by $L'_m$. 

\begin{corollary}
The set $\ev(\cB')$ generates $\ev(\tr(\cL(\infty Q)))$ over $\fq$.
The set $\ev(\cB'_m)$ generates $\tr(C(\cX,D,mQ))$ over $\fq$.
\end{corollary}

In order to ensure the self-orthogonality of trace codes, we shall use the following fact.

\begin{lemma}\label{trtr}
Let $\bbx,\bby\in\fqr^n$. We have $\langle\tr(\bbx),\tr(\bby)\rangle=\tr(\langle\bbx,\tr(\bby)\rangle)$.
\end{lemma}
\begin{proof}
For $\bbx,\bbz\in\fqr^n$ and $0\le i\le  r$ we have $\langle \bbx^{q^i},\bbz\rangle=\langle\bbx, \bbz^{q^{r-i}} \rangle^{q^i}$. Take $\bbz=\tr(\bby)$. Note that $\bbz\in\fq^n$ and hence $\bbz^{q^{r-i}}=\bbz$. Thus
\begin{equation}
\langle\tr(\bbx),\tr(\bby)\rangle=\sum_{i=0}^{r-1} \langle \bbx^{q^i}, \tr(\bby)\rangle = 
\sum_{i=0}^{r-1} \langle \bbx, \tr(\bby)\rangle^{q^i} = \tr( \langle\bbx,\tr(\bby)\rangle) .
\end{equation}
\end{proof}

Remember that given an integer $m$,¡ with $0\leq m\leq n+2g-2$, we write  $m^{\perp}=n+2g-2-m$. The self-orthogonality and the parameters of trace codes can be checked by using the following result.

\begin{proposition}\label{traces}
Let $(\cX,Q)$ be a pointed curve of genus $g$ over $\fqr$ verifying the self-dual property. 
\begin{enumerate} 
\item[(1)] If $mq^{\lfloor r/2 \rfloor}\le m^{\perp}$ then  $\tr(C(\cX,D,mQ))$ is self-orthogonal over $\fq$.
\item[(2)]  If $mq^{r-1}< n$ then $\ev(\tr(f))\neq 0$ for all $f\in\cL(mQ)$.
\item[(3)] $d(\tr(C(\cX,D,mQ))^{\perp})=d(C(\cX,D,m^{\perp}Q)|\fq)\ge d(C(\cX,D,m^{\perp}Q))$.
\end{enumerate}
\end{proposition}
\begin{proof}
(1) Assume $r$ is even.  If $mq^{ r/2 }\le m^{\perp}$ then for all $i=0,\dots,r/2$, we have 
\begin{equation}
C(\cX,D,mQ)^{q^i}\subseteq  C(\cX,D,mq^iQ)\subseteq C(\cX,D,mQ)^{\perp}.
\end{equation}
Let $f,g\in\cL(mQ)$. According to Lemma \ref{trtr} we have
\begin{eqnarray}
\langle \tr(\ev(f)), \tr(\ev(g))\rangle &=&\tr(\sum_{j=0}^{r-1} \langle \ev(f), \ev(g)^{q^j}\rangle) \\
 & \hspace*{-3cm} = & \hspace*{-1.5cm}\tr (\sum_{j=0}^{r/2-1} \langle \ev(f), \ev(g)^{q^j}\rangle ) +  \tr ( \sum_{j=r/2}^{r-1} \langle \ev(f), \ev(g)^{q^j}\rangle).
\end{eqnarray}
By eq. (21) all summands in the first sum of eq. (23) are zero and so its trace is zero. For the second sum note that as seen in the proof of Lemma \ref{trtr} we can write $\langle \ev(f), \ev(g)^{q^j}\rangle=\langle \ev(f)^{q^{r-j}}, \ev(g)\rangle^{q^j}$, with $r-j\le r/2$. Then $\langle \ev(f)^{q^{r-j}}, \ev(g)\rangle =0$ and the second trace is also 0. The case $r$ odd is similar.
(2) Let  $f\in\cL(mQ)$. It holds that $v(\tr(f))\le mq^{r-1}<n$, hence   $\ev(\tr(f))\neq \bbcero$.
(3) Follows from Delsarte's theorem.
\end{proof}

\begin{example}
Consider the Suzuki curve with $65$ rational points and genus $14$ over ${\mathbb F}_8$. According to  Proposition \ref{traces}(1),  the traces   $\tr(C(\cX,D,mQ))$ over  ${\mathbb F}_2$ are self-orthogonal for $0\le m\le 30$. Furthermore  since  $\dim(\tr(C(\cX,D,30Q)))=32$, this code is self-dual. Thus the bound stated in item (1) is sharp in this case. From these traces, we derive quantum codes with  parameters 
$[[64,62,2]]_2^{\ddag}$,
$[[64,50,4]]_2^{\ddag}$,
where the notation $[[n,k,d]]_q^{\dag\ddag}$ is as in Section \ref{sec:11}.
\end{example}

Let us consider the case $r=2$ and let $C=C(\cX,D,mQ)$  be a code over  ${\mathbb F}_{q^2}$. Note that the condition on $m$ stated in Corollary \ref{cor} to ensure  the Hermitian self-orthogonality of  $C$, is exactly the same as the condition of Proposition \ref{traces}(1), to ensure the self-orthogonality of $\tr(C)$. Thus, in some sense, the construction of quantum codes from traces  of self-orthogonal codes, extends to $r>2$ the construction of quantum codes from Hermitan self-orthogonal codes over $\fqd$. Note also that besides considering traces  $\tr(C(\cX,D,mQ))$, we can use 'incomplete' traces, that is subcodes generated by some (but not all) elements of  $\ev(\cB'_m)$. In many cases these subcodes  can provide quantum codes with better parameters than trace codes themselves.

\begin{example}
Consider the elliptic Hermitian curve $\cX:y^2+y=x^3$  with $9$ rational points over ${\mathbb F}_4$ and let $\alpha$ be a primitive element of ${\mathbb F}_4$ over ${\mathbb F}_2$. The trace code $\tr(C(\cX,D,3Q)) = \langle \bbuno,\ev(\tr(x)),\ev(\tr(\alpha x)), \ev(\tr(y)), \ev(\tr(\alpha y))\rangle$ has dimension $5$ and dual distance $4$. Thus it cannot be self-orthogonal. Remove the generator $\ev(\tr(y))$. Then we get a self-dual code of dimension $4$ whose dual distance is also 4, that gives a quantum $[[8,0,4]]_2$ code whose parameters cannot be improved \cite{tablas2}. Furthermore
this code cannot be obtained from $\cX$ by using Hermitian self-orthogonality and construction (A). 
In the same way, consider the Hermitian curve over ${\mathbb F}_{q^{2r}}$ with $q^{3r}+1$ points. The dual distance of the code $C(\cX,D,(q^r+1)Q)$ can be computed by using the order bound. 
Consider the trace $\tr(C(\cX,D,(q^r+1)Q))$ of this code over $\fq$ of dimension $2r+1$ and  remove at most $r-1$ generators from $\ev(\tr(y)),\dots,\ev(\tr(\alpha^{r-1} y))$ in order to obtain a self-orthogonal code with the same dual distance as $\tr(C(\cX,D,(q^r+1)Q))$.  
Direct computations show that we get quantum codes 
$[[64,50 ,4]]_2^{\ddag}$,
$[[512,492 ,4]]_2^{\ddag}$,
$[[27,19 ,3]]_3^{\dag}$,
$[[729,715 ,3]]_3^{\dag}$,
$[[64,56 ,3]]_4^{\dag}$,
$[[125,117 ,3]]_5^{\dag}$,
$[[343,335 ,3]]_7^{\dag}$,
$[[512,504 ,3]]_8^{\dag}$,
$[[729,721 ,3]]_9^{\dag}$, 
reaching good parameters in all cases. The  same procedure can be carried out by using codes  from other curves. For example, from the Norm-Trace curve over $\fqr$  we obtain quantum codes 
$[[32,20 ,4]]_2^{\ddag}$,
$[[128,112 ,4]]_2^{\ddag}$, etc.
\end{example}

\section*{acknowledgements}
The first author was supported by Spanish Ministerio de Econom\'{\i}a y Competitividad 
under grant MTM2015-65764-C3-1-P MINECO/FEDER.
The second and third authors were supported by CNPq-Brazil under grants 201584/2015-8 and 308326/2014-8 respectively.
This paper was written partially during a visit of the second and third authors to the University of Valladolid. Both wish to thank the institution for the hospitality and  support. The authors are grateful to the anonymous referees for
their invaluable comments and suggestions which substantially helped improving the quality of this paper.


\begin{thebibliography}{99}


\bibitem{AK}
Ashikhim, A.,  Knill, E.:
Non-binary quantum stabilizer codes. 
IEEE Transactions on Information Theory  {\bf 47}, 3065-3072 (2001).

\bibitem{tablas}
Edel, Y.:
Some good quantum twisted codes.
Online available at 
http://www.mathi.uni-heidelberg.de/$\thicksim$yves/Matrizen/QTBCH/QTBCHindex.html

\bibitem{CS}
Calderbank, A.R., Shor, P.W.:
Good quantum error-correcting codes exist.
Physical Review A {\bf 54}, 1098-1105 (1996).

\bibitem{Chen}
Chen, H.: 
Some good quantum error-correcting codes from algebraic geometry codes.
IEEE Transactions on Information Theory {\bf 47},  2059-2061  (2001).

\bibitem{CKT} 
Cossidente, A., Korchm\' aros, G.,  Torres, F.: 
On curves covered by the Hermitian curve.
J. Algebra {\bf 216}(1),  56-76 (1999).

\bibitem{GVB}
Feng, K.,   Ma, Z.:
A Finite Gilbert-Varshamov Bound for Pure Stabilizer Quantum Codes.
IEEE Transactions on Information Theory {\bf 50},  3323-3325  (2004).

\bibitem{GH}
Galindo, C.,  Hernando, F.:
Quantum codes from affine variety codes and their subfield-subcodes.
Designs, Codes and Cryptography {\bf 76}(1), 89-100 (2015). 

\bibitem{GHR}
Galindo, C., Hernando, F., Ruano, D.:
Stabilizer quantum codes from J-affine variety codes and a new Steane-like enlargement.
Quantum Information Processing {\bf 14}(9), 3211-3231 (2015). 

\bibitem{GMRT}
Geil, O., Munuera, C., Ruano, D.,  Torres, F.:
On the order bounds for one point codes.
Advances in Mathematics of Communications {\bf 5}(3),  489-504 (2011).

\bibitem{tablas2}
Grassl, M.: 
Bounds on the minimum distance of linear codes and quantum codes.
Online available at http://www.codetables.de.

\bibitem{HKT} 
Hirschfeld, J., Korchm\' aros, G., Torres, F.: 
Algebraic curves over a finite field. 
Princeton Series in Applied Mathematics. Princeton University Press, Princeton (2008).

\bibitem{HLP}
H{\o}holdt, T., van Lint, J.H.,  Pellikaan, R.:
Algebraic Geometry codes. 
In: Handbook of Coding Theory, vol. 1, 871-961. Pless, V., Huffman, W.C. (Eds.). Elsevier, Amsterdam,  (1998).

\bibitem{Ji}
Jin, L.:
Quantum stabilizer codes from maximal curves.
IEEE Transactions on Information Theory {\bf 60}(1),  313-316 (2014).

\bibitem{JX}
Jin, L., Xing, C.P.:
Euclidean and Hermitian self-orthogonal Algebraic Geometry codes and their application to Quantum codes. 
IEEE Transactions on Information Theory {\bf 58}(8),  5484-5489 (2012).

\bibitem{KM}  
Kim, J., Mathews, G.L.:
Quantum error-correcting codes from algebraic curves.
In: Advances in Algebraic Geometry codes, 419-444. Martinez, E., Munuera, C., Ruano, D. (Eds.).
Word Scientific, Hackensack (2008).

\bibitem{KW}
Kim, J.,  Walker, J.:
Nonbinary quantum error-correcting cods from algebraic curves.
Discrete Mathematics {\bf 308},  3115-3124 (2008).

\bibitem{OM}
Olaya, W., Munuera, C.:
An introduction to Algebraic Geometry codes. 
In: Algebra for Secure and Reliable Communication Modeling,  87-117. Lahyane, M., Martinez, E. (Eds.). Contemporary Mathematics 642, American Mathematical Society, Providence (2015).

\bibitem{magma}
Magma Computational Algebra System. Online available at http://magma.maths. usyd.edu.au/magma/.

\bibitem{Castle2}
Munuera, C., Sepulveda, A., Torres, F.:
Castle curves and codes. 
Advances in Mathematics of Communications {\bf 3}, 399-408 (2009).

\bibitem{SK}
Sarpevalli, P.K., Klappenecker, A.:
Nonbinary quantum codes from Hermitian curves.
In: Applied Algebra, Algebraic Algorithms and Error-Correcting Codes, 136-143.
Lecture Notes in Computer Science 3857, Springer, Berlin (2006).

\bibitem{Sh}
Shaska, T.: 
Quantum codes from algebraic curves with automorphisms.
Condensed Matter Physics {\bf 11}(2) (2008), 383-396 (2008).

\bibitem{Ste}
Steane, A.M.:
Multiple-particle interference and quantum error correction.
Proceedings of the Royal Society of London-Series A {\bf 452}, 2551-2557 (1996).

\bibitem{SD}
Stichtenoth, H.:
Self-dual Goppa codes.
Journal of Pure and Applied Algebra {\bf 55}(1-2),  199-211  (1988).

\bibitem{Sti}
Stichtenoth, H.:
Algebraic Function Fields and Codes.
Springer-Verlag, Berlin (1993). 

\end{thebibliography}
\end{document}